\documentclass[letterpaper, 10 pt, conference]{ieeeconf}
\usepackage[utf8]{inputenc}
\usepackage{amsmath}
\usepackage{amssymb}
\usepackage{graphicx}
\usepackage{caption}
\usepackage{babel}
\usepackage{balance}
\usepackage{cite}
\usepackage[position=bottom]{subfig}
\usepackage{floatrow}

\IEEEoverridecommandlockouts

\let\theoremstyle\relax
\usepackage{amsthm}

\theoremstyle{definition}
\newtheorem{definition}{Definition}
\newtheorem{theorem}{Theorem}

\newtheorem{assumption}{Assumption}
\newtheorem{proposition}{Proposition}
\newtheorem{problem}{Problem}

\usepackage[ruled]{algorithm2e}

\usepackage{todonotes}

\title{Conference2021}
\title{\LARGE \bf Energy-Optimal Goal Assignment of Multi-Agent System\\with Goal Trajectories in Polynomials}

\author{Heeseung Bang, \textit{Student Member, IEEE}, Logan E. Beaver, \textit{Student Member, IEEE},\\ Andreas A. Malikopoulos, \textit{Senior Member, IEEE}
    \thanks{This research was supported by the Sociotechnical Systems Center (SSC) at the University of Delaware.}
	\thanks{The authors are with the Department of Mechanical Engineering, University of Delaware, Newark, DE 19716, USA. (emails: \tt\small{heeseung@udel.edu};  \tt\small{lebeaver@udel.edu};  \tt\small{andreas@udel.edu}.)}
}

\date{October 2020}

\begin{document}

\maketitle

\begin{abstract}
    In this paper, we propose an approach for solving an energy-optimal goal assignment problem to generate the desired formation in multi-agent systems. Each agent solves a decentralized optimization problem with only local information about its neighboring agents and the goals.
    The optimization problem consists of two sub-problems. The first problem seeks to minimize the energy for each agent to reach certain goals, while the second problem  entreats an optimal combination of goal and agent pairs that minimizes the energy cost.
    By assuming the goal trajectories are given in a polynomial form, we prove the solution to the formulated problem exists globally. Finally, the effectiveness of the proposed approach is validated through the simulation.
\end{abstract}

\section{Introduction}

    Control of swarm systems is an emerging topic in the fields of controls and robotics.
    Due to their adaptability and flexibility \cite{Oh2017}, swarm systems have attracted considerable attention in transportation \cite{chalaki2020experimental}, construction \cite{Lindsey2012ConstructionTeams}, and surveillance \cite{Corts2009} applications.
    As we deploy swarms in experimental testbeds \cite{Pickem2017TheTestbed,Malikopoulos2018b, Rubenstein2012, Beaver2020DemonstrationCity} and outdoor experiments \cite{Vasarhelyi2018OptimizedEnvironments}, it is critical to minimize the cost per agent to ensure swarms are an affordable solution to emerging problems.
    This is the driving force behind energy-optimal control algorithms, which reduce the battery storage requirements, and therefore, the cost, of agents while simultaneously expanding their useful life.
    
    A fundamental problem in swarm systems is the assignment of agents to a particular formation.
    There is a rich literature on the creation of a desired formation, such generating rigid formations from triangular sub-structures \cite{Guo2010,Hanada2007}, crystal growth-inspired algorithms \cite{Song}, and region-based formation controllers \cite{Cheah2009}.
    It is also possible for agents to construct formations using only scalar, bearing, or distance measurements \cite{Swartling2014,Lin2004}, and many formation problems may be solved using consensus techniques \cite{Olfati-Saber2007}.
    However, only a few of these approaches consider the energy cost to individual agents in the swarm.
    
    Similar to the efforts reported in \cite{Turpin2014, Turpin2013GoalRobots, Morgan2016}, we seek the assignment of a finite number of agents to a set of desired states.
    Our approach leverages optimal control to guarantee inter-agent collision avoidance while minimizing the energy consumed by each agent.
    Unlike \cite{Turpin2014}, our approach is not pairwise between agents, instead we consider all nearby agents during goal assignment.
    Our approach also does not require the agents to be assigned to unique goals a priori.
    Similar to \cite{Turpin2013GoalRobots}, our approach imposes a priority ordering on the agents to generate assignments and trajectories.
    However, our approach to prioritization is dynamic and decentralized, as opposed to the global static priority presented in \cite{Turpin2013GoalRobots}. 
    Finally, our approach to assignment only considers the local area around an agent, unlike the global auction algorithm in \cite{Morgan2016}.
    Additionally, we consider the unconstrained energy cost required to reach a goal during assignment, whereas \cite{Turpin2014, Turpin2013GoalRobots,Morgan2016} only consider the distance to the goal.
    In other words, our approach considers the energy cost required for the agent to match the goal's velocity.
    
    By leveraging optimal control, we explicitly allow for the prioritization of safety as a hard constraint on the system.
    Strong guarantees on safety are valuable to avoid inter-agent collisions and to guarantee that agents avoid obstacles in the environment.
    We propose an extension of our previous work on energy-optimal goal assignment and trajectory generation \cite{Beaver2019AGeneration,Beaver2020AnAgents}.
    The main contributions of this paper are:
    (1) we optimally determine the arrival time of each agent during assignment, while we provide a set of sufficient conditions on the goal dynamics to guarantee that the arrival time is finite; and
    (2) we propose an event-triggered approach to goal assignment that guarantees all agents will converge to a unique goal.
    We also provide a numerical demonstration of our improved assignment and trajectory generation scheme.
    

The remainder of the paper is organized as follows. 
In Section \ref{sec:ModelingFramework}, we formulate the optimal goal assignment and trajectory generation problem. 
In Section \ref{sec:GoalAssignment}, we formulate the goal assignment problem and provide an event-triggered update scheme that guarantees convergence.
In Section \ref{sec:PathPlanning}, we explain the trajectory planning scheme, and in Section \ref{sec:Simulation}, we quantify the improvement in performance over our previous work \cite{Beaver2019AGeneration, Beaver2020AnAgents}.
Finally, we draw our conclusions and propose future research directions in Section \ref{sec:Conclusion}.

\section{Modeling Framework} \label{sec:ModelingFramework}

We consider a problem of generating a desired formation by allocating $N\in\mathbb{N}$ agents into $M\in\mathbb{N}$ goals, where $M \geq N$. The agents and the goals are indexed by the sets $\mathcal{A}=\{1,\dots,N\}$ and $\mathcal{F}=\{1,\dots,M\}$, respectively. For continuous time $t\in\mathbb{R}_{\geq 0}$, each agent $i\in\mathcal{A}$ obeys double-integrator dynamics,
\begin{align}
    \dot{\mathbf{p}}_{i}(t)&=\mathbf{v}_i(t), \label{eqn:pDynamics} \\
    \dot{\mathbf{v}}_{i}(t)&=\mathbf{u}_i(t), \label{eqn:vDynamics}
\end{align}
where $\mathbf{p}_i(t)\in\mathbb{R}^2$ and $\mathbf{v}_i(t)\in\mathbb{R}^2$ are the time-varying position and velocity vectors, and $\mathbf{u}_i(t)\in\mathbb{R}^2$ is the control input. The control input and velocity of each agent are bounded by
\begin{align}
    ||\mathbf{v}_i(t)|| &\leq v_{\max}, \label{eqn:vBounds}\\
    ||\mathbf{u}_i(t)|| &\leq u_{\max}, \label{eqn:uBounds}
\end{align}
where $v_{\max}$ and $u_{\max}$ are the maximum allowable speed and control inputs, and $|| \cdot ||$ is the Euclidean norm. The state of each agent is given by the time-varying vector
\begin{equation}
    \mathbf{x}_i(t) = \left[\begin{array}{c}
         \mathbf{p}_i(t) \\ \mathbf{v}_i(t)
    \end{array} \right].
\end{equation}
We denote the distance between two agents $i, j\in\mathcal{A}$ by 
\begin{equation}
	d_{ij}(t) = ||\mathbf{p}_i(t)-\mathbf{p}_j(t)||.
\end{equation}
In order to avoid collisions between agents, we impose the following pairwise constraints for all agents $i,j\in\mathcal{A}, i\neq j$,
\begin{align}
	d_{ij}(t) \geq 2R,&\quad\forall t\geq 0,\\
    h\gg 2R,&
\end{align}
where $R \in \mathbb{R}_{>0}$ is the radius of a safety disk centered on each agent, and $h \in \mathbb{R}_{>0}$ is the sensing and communication horizon. Next, we define the neighborhood of an agent, which is our basis for local information. 
\begin{definition} \label{def:neighborhood}
	The \textit{neighborhood} of agent $i\in\mathcal{A}$ is the time-varying set 
	\begin{equation*}
	\mathcal{N}_i(t) = \Big\{j\in\mathcal{A} ~ \Big| ~ d_{ij}(t) \leq h\Big\}.
	\end{equation*}
	Agent $i$ may sense and communicate with every neighboring agent $j\in\mathcal{N}_i(t)$.
\end{definition}

We also define the notion of \textit{desired formation}.
\begin{definition}
	The \textit{desired formation} is the set of time-varying vectors $\mathcal{G}(t) = \{\mathbf{p}_k^*(t) \in \mathbb{R}^2 ~ | ~ k \in \mathcal{F}\}$. 
\end{definition}
The set $\mathcal{G}(t)$ can be prescribed offline, i.e., by a designer, or online by a high-level planner.
Since we consider the desired formation with polynomial trajectories, each goal $k\in\mathcal{F}$ has the form
\begin{equation}
    \mathbf{p}_k^*(t) = \sum_{l=0}^{\eta} \mathbf{c}_{k,l} t^l, \quad \eta \geq 2, \label{eqn:goal}
\end{equation}
where $\eta$ is the degree of the polynomial and the coefficients $\mathbf{c}_{k,l}\in\mathbb{R}^2$ are constant vectors.

We impose the following model for the rate of energy consumption by agent $i\in\mathcal{A}$,
\begin{equation}
	\dot{E}_i(t)=\frac{1}{2}||\mathbf{u}_i(t)||^2.   \label{eqn:e_dot}
\end{equation}
Physically, this energy model implies that minimizing $L^2$ norm of acceleration directly reduces the total energy consumed by each agent.

In our modeling framework, we impose the following assumptions.

\begin{assumption} \label{asm:asm1}
    There are no errors or delays with respect to communication and sensing within each agent's neighborhood. 
\end{assumption}

\begin{assumption}\label{asm:asm2}
    The energy cost of communication is negligible, i.e., the energy consumption is only in the form of \eqref{eqn:e_dot}.
\end{assumption}

\begin{assumption}\label{asm:asm3}
   Each agent has a low-level onboard controller that can track the generated optimal trajectory.
\end{assumption}

Assumption \ref{asm:asm1} is employed to characterize the idealized performance of our approach. This may be relaxed by using a stochastic optimal control problem, or robust control, for trajectory generation. Assumption \ref{asm:asm2} may be relaxed for the case with long-distance communication. For that case, the communication cost can be controlled by varying the communication horizon $h$.
Assumption \ref{asm:asm3} may be strong for certain applications. This assumption may be relaxed by including kinematic constraints in the optimal trajectory generation problem, or by employing a robust low-level controller, such as a control barrier function, for tracking.

\section{Optimal Goal Assignment} \label{sec:GoalAssignment}
The objective of a goal assignment problem is to assign each agent to a unique goal such that the total energy consumption of all agents is minimized. We separate this into two sub-problems: (1) finding the minimum-energy unconstrained trajectory for each agent to reach every goal, and (2) finding the optimal assignment of agents to goals such that total energy consumption is minimized and at most one agent is assigned to each goal. 

To solve the first sub-problem, we consider the case of any agent $i\in\mathcal{A}$ traveling between two fixed states with the energy model in the form of \eqref{eqn:e_dot}. In this case, Hamiltonian analysis yields the following optimal unconstrained minimum-energy trajectory \cite{Malikopoulos2018a},
\begin{align}
    \mathbf{u}_i(t) &= \mathbf{a}_i t + \mathbf{b}_i, \label{eqn:uUnc}\\
    \mathbf{v}_i(t) &= \frac{\mathbf{a}_i}{2}t^2 + \mathbf{b}_i t + \mathbf{c}_i, \label{eqn:vUnc}\\
    \mathbf{p}_i(t) &= \frac{\mathbf{a}_i}{6}t^3 + \frac{\mathbf{b}_i}{2} t^2 + \mathbf{c}_i t + \mathbf{d}_i \label{eqn:pUnc},
\end{align}
where $\mathbf{a}_i$, $\mathbf{b}_i$, $\mathbf{c}_i$, and $\mathbf{d}_i$ are constant vectors of integration.
Thus, we get the minimum required total-energy for agent $i$ to reach the goal $k\in\mathcal{F}$, by substituting \eqref{eqn:uUnc} into \eqref{eqn:e_dot}, that is,
\begin{align}
E_{i,k} (t_{i,k}) &= \int_0^{t_{i,k}} ||\mathbf{u}_i (\tau)||^2 d\tau \nonumber\\
&= \frac{a_{i,x}^2+a_{i,y}^2}{3}t_{i,k}^3 +(a_{i,x}b_{i,x}+a_{i,y}b_{i,y})t_{i,k}^2 \nonumber \\
&~+(b_{i,x}^2+b_{i,y}^2)t_{i,k}, \label{eqn:E}
\end{align}
where $t_{i,k}$ is the time taken for the agent $i$ to reach the goal $k$, and $\mathbf{a}_i = [a_{i,x},a_{i,y}]^T$, $\mathbf{b}_i = [b_{i,x},b_{i,y}]^T$ are the coefficients of \eqref{eqn:uUnc}.
We solve for the coefficients $\mathbf{a}_i$ and $\mathbf{b}_i$ by substituting the boundary conditions into \eqref{eqn:vUnc} and \eqref{eqn:pUnc},
\begin{align}
\mathbf{a}_i &= \frac{12}{t_{i,k}^3} \left( \mathbf{p}_{i,0} - \mathbf{p}_k^*(t_{i,k}) \right) + \frac{6}{t_{i,k}^2} \left(\mathbf{v}_{i,0}+\mathbf{v}_k^*(t_{i,k}) \right), \label{eqn:ai}\\
\mathbf{b}_i &= -\frac{6}{t_{i,k}^2} \left(\mathbf{p}_{i,0} - \mathbf{p}_k^*(t_{i,k}) \right) - \frac{2}{t_{i,k}} \left(2 \mathbf{v}_{i,0}+\mathbf{v}_k^*(t_{i,k}) \right). \label{eqn:bi}
\end{align}
Here, $\mathbf{p}_{i,0}$ and $\mathbf{v}_{i,0}$ are the initial position and velocity of the agent $i$, respectively.
Next, we define an optimization problem to find the minimum-energy arrival time.
\begin{problem}[Energy Minimization] \label{prb:energy}
    The minimum-energy arrival time for agent $i\in\mathcal{A}$ traveling to goal $k\in\mathcal{F}$ is found by solving the following optimization problem,
	\begin{align}
	    E_{i,k}^* =& \min_{t_{i,k}} E_{i,k} (t_{i,k})\\
	        &\text{subject to } \eqref{eqn:goal},\nonumber
	\end{align}
\end{problem}

\begin{proposition} \label{prp:minE}
    For goal trajectories in the form of \eqref{eqn:goal}, there always exists a globally optimal solution to Problem \ref{prb:energy}.
\end{proposition}
\begin{proof}
    First we substitute \eqref{eqn:goal} and its time derivative into \eqref{eqn:ai} and \eqref{eqn:bi}, which yields equations of the form
    \begin{align} \label{eqn:minEa}
        \mathbf{a}_i = \sum_{l=0}^{\eta} \mathbf{c}_{l,a} t^{l-3}, \\
        \mathbf{b}_i = \sum_{l=0}^{\eta} \mathbf{c}_{l,b} t^{l-2}, \label{eqn:minEb}
    \end{align}
    Squaring \eqref{eqn:minEa} and \eqref{eqn:minEb} and substituting the result into \eqref{eqn:E} yields an equation of the form
%
    \begin{equation}
    	E_{i,k}(t_{i,k}) = \sum_{l=0}^{2\eta} \alpha_l t_{i,k}^{l-3}, \label{eqn:Epol}
    \end{equation}
    where $\alpha_l$ are constant numbers, and $\alpha_{2\eta}>0$, $\alpha_{0}>0$. 
    Eq. \eqref{eqn:goal} implies that $\eta \geq 2$, thus \eqref{eqn:Epol} always has polynomial and inverted radical terms.
    Thus, as $t\to\infty$, the polynomial terms dominate and
    \begin{equation}
        \lim_{t\to\infty}E_{i,k}(t)=\infty. \label{eqn:lim1}
    \end{equation}
    As $t\to 0^+$, the inverted radical terms dominate, and
    \begin{equation}
        \lim_{t\to0^{+}}E_{i,k}(t)=\infty. \label{eqn:lim2}
    \end{equation}
    Finally, $\mathbf{u}_i(t)\in\mathbb{R}^2$ implies that $E_{i,k}(t) \geq 0$ for $t\in(0, \infty)$ by \eqref{eqn:E}.
    %
    %
    From \eqref{eqn:lim1}, if we select sufficiently small positive number $\varepsilon$, there exists $\gamma$ such that $E_{i,k}(\gamma) > E_{i,k}(\varepsilon)$, $\forall \gamma \in (0,\varepsilon)$. Likewise, from \eqref{eqn:lim2}, for sufficiently large number $\beta$, there exists $\delta$ such that $E_{i,k}(\beta) < E_{i,k}(\delta)$, $\forall \delta \in (\beta,\infty)$. This implies that the local minimum in $[\varepsilon,\beta]$ is the global minimum as well. According to the boundness theorem in calculus, a continuous function in the closed interval is bounded on that interval. That is, for the continuous function \eqref{eqn:Epol} in $[\varepsilon,\beta]$, there exist real number $\underline{m}$ and $\bar{m}$ such that:
    \begin{equation}
        \underline{m} < E_{i,k}(t) < \bar{m},~\forall t \in [\varepsilon,\beta],
    \end{equation}
    and the proof is complete.

\end{proof}

Proposition \ref{prb:energy} enables the agent to consider the energy-optimal arrival time during goal assignment.
In contrast, our previous work \cite{Beaver2019AGeneration,Beaver2020AnAgents} uses a fixed arrival time that is selected offline by a designer.


After the energy minimization is complete, each agent assigns itself and its neighbors to unique goals. 
This is achieved using an assignment matrix $\mathbf{A}_i(t)$ of size $|\mathcal{N}_i(t)|\times M$, which we define next.

\begin{definition} \label{def:assignmatrix}
    The \textit{assignment matrix} $\mathbf{A}_i(t)$ for each agent $i\in\mathcal{A}$ maps all agents $j\in\mathcal{N}_i(t)$ to a unique goal index $g\in\mathcal{F}$. The elements of $\mathbf{A}_i(t)$ are binary valued, and each agent is assigned to exactly one goal.
\end{definition}

We determine the assignment matrix by solving a decentralized optimization problem, which we present later in this section.
Next, we define the prescribed goal to show how the agent uses the assignment matrix.

\begin{definition} \label{def:prescribedgoal}
    For agent $i\in\mathcal{A}$, the \textit{prescribed goal} is
    \begin{equation}
        \mathbf{p}_i^a(t) \in \big\{ \mathbf{p}_k^*\in\mathcal{G}~|~a_{ik}=1,a_{ik}\in\mathbf{A}_i(t), k\in\mathcal{F} \big\} .
    \end{equation}
\end{definition}

Since the prescribed goal is determined using only local information, it is possible that two agents with different neighborhoods will prescribe themselves the same goal.
To solve this problem, each agent must know which agent it is competing with and which one has priority for the goal. This motivates our definitions of competing agents and the priority indicator function. 

\begin{definition} \label{def:competingagents}
    The set of \textit{competing agents} for agent $i\in\mathcal{A}$ is given by
    \begin{equation}
        \mathcal{C}_i(t) = \left\{ j\in\mathcal{N}_i(t) ~\large|~\mathbf{p}_j^a(t) = \mathbf{p}_i^a(t) , ~i\neq j \right\}.
    \end{equation}
\end{definition}

The information about competing agent is updated whenever a new agent enters the neighborhood of agent $i$.
If there is at least one competing agent, that is $|\mathcal{C}_i(t)| \geq 1$, then all agents $j\in\mathcal{C}_i(t)$ must compare their priority indicator function, which we define next.

\begin{definition} \label{def:indicatorfunc}
    For each agent $i\in\mathcal{A}$, we define the \textit{priority indicator function} $\mathbb{I}_i : \mathcal{A}\setminus\{i\}\to\{0,1\}$. 
    We say that that agent $i\in\mathcal{A}$ has priority over agent $j\in\mathcal{A}\setminus\{i\}$, if and only if $\mathbb{I}_i(j)=1$.
    Additionally, $\mathbb{I}_i(j)=1$ if and only if $\mathbb{I}_j(i)=0$.
\end{definition}

The functional form of the priority indicator function is determined offline by a designer and is the same for all agents. By Assumption \ref{asm:asm1} the information required to evaluate priority is instantaneously and noiselessly measured and communicated between agents.
Following this policy, the agent with no priority is permanently banned from its prescribed goal.
\begin{definition} \label{def:bannedgoal}
We denote the set of \textit{banned goals} for agent $i\in\mathcal{A}$ as
\begin{equation}
    \mathcal{B}_i(t) \subset \mathcal{F}.
\end{equation}
Elements are never removed from $\mathcal{B}_i(t)$, and a goal $g\in\mathcal{F}$ is added to $\mathcal{B}_i(t)$, if $\mathbf{p}_i^a(t) = \mathbf{p}_g^*(t)\in\mathcal{G}$ and $\mathbb{I}_i(j) = 0$ for any $j\in\mathcal{C}_i(t) \setminus \{ i \}$.
\end{definition}

Agent $i\in\mathcal{A}$ assigns itself a prescribed goal by 
solving the following optimization problem, where we include the banned goals as constraints.

\begin{problem}[Goal Assignment] \label{prb:assignment}
	Each agent $i\in\mathcal{A}$ selects its prescribed goal (Definition \ref{def:prescribedgoal}) by solving the following binary program:
	\begin{align}
	\underset{a_{jk}\in\mathbf{A}_i}{\text{min}} \Bigg\{
	\sum_{j\in\mathcal{N}_i(t)} \sum_{k\in\mathcal{F}} a_{jk} E_{j,k}^*\Bigg\}
	\end{align}
	subject to:
	\begin{align}
			 \sum_{k\in\mathcal{F}}~ a_{jk} &= 1, \quad j\in\mathcal{N}_i(t), \label{eqn:p11} \\
			 \sum_{j\in\mathcal{N}_i(t)} a_{jk} &\leq 1, \quad k\in\mathcal{F},\label{eqn:p12}\\
			 a_{jk} &= 0, \quad \forall~j\in\mathcal{N}_i(t),~ k\in\mathcal{B}_j(t),\label{eqn:p13} \\
			 a_{jk} &\in \{0, 1\} \notag.
	\end{align}
\end{problem}

Next, we present Algorithm \ref{alg:algorithm}, which describes our event-driven protocol for assigning agents to goals using the competing agent set, priority indicator function, and banned goal set. 
\hspace*{-2cm}
\begin{algorithm} [ht]
    Solve Problem \ref{prb:assignment}\;
    Determine prescribed goal\;
    Generate optimal trajectory to assigned goal\;
    \If{$|\mathcal{C}_i(t)| \geq 1$}
    {
    Compare $\mathbb{I}_i(j)$ for all $j\in\mathcal{C}_i(t)$\;
        \If{any $\mathbb{I}_i(j) = 0$}
        {Add current goal to $\mathcal{B}_i(t)$\;
        Solve Problem \ref{prb:assignment}\;
        Determine prescribed goal\;
        Generate optimal trajectory to assigned goal\; }
    }
    \caption{Event-driven algorithm to determine the prescribed goal for each agent $i\in\mathcal{A}$.} \label{alg:algorithm}
\end{algorithm}

\begin{proposition}[Solution Existence] \label{prp:solutionExistance}
    A solution to Problem \ref{prb:assignment} always exists.
%
%
\end{proposition}

\begin{proof}

    Let $\mathcal{B}(t) = \bigcup_{i\in\mathcal{A}} \mathcal{B}_i(t)$ be the set of all goals which any agent is banned from.
    Let $n_b(t) = |\mathcal{B}(t)|$, then based on Algorithm \ref{alg:algorithm}, there must be exactly $n_b(t)$ agents assigned to the $n_b(t)$ banned goals.
    Thus, any agent $i\in\mathcal{A}$ must assign at most $N - n_b(t)$ agents to $M - n_b(t)$ goals when solving Problem \ref{prb:assignment}.
    As $M \geq N$, $M - n_b(t) \geq N - n_b(t)$, and the feasible space of Problem \ref{prb:assignment} is always non-empty.

    \end{proof}

%
%

Each agent $i\in\mathcal{A}$ initially solves Problem \ref{prb:assignment} to assign itself to a goal, and re-solves Problem \ref{prb:assignment} whenever its neighborhood $\mathcal{N}_i(t)$ switches and the set of competing agents becomes non-empty.
It is possible that several agents may assign themselves to the same goal.
If it is the case, all conflicting agents repeat the banning and assignment process until all agents are assigned to a unique goal.
Next, using Proposition \ref{prp:minE} and Proposition \ref{prp:solutionExistance}, we propose Theorem \ref{thm:theorem} which guarantees convergence of all agents to a unique goal in a finite time.

\begin{theorem} \label{thm:theorem}
Let any agent $i\in\mathcal{A}$ be assigned to a goal $k\in\mathcal{F}$ under our proposed banning and reassignment approach (Definitions \ref{def:competingagents} - \ref{def:bannedgoal}) and polynomial goal trajectories \eqref{eqn:goal}.
If the solution to Problem \ref{prb:energy} is never increasing, i.e., $ E_{i,k}^*(t_1) \geq E_{i,k}^*(t_2)$ for sequential assignments of agent $i$ to goal $k$ at times $t_1, t_2 \in \mathbb{R}_{\geq0},$ where $t_2 > t_1$, then all agents arrive at their unique assigned goal in finite time.
\end{theorem}

\begin{proof}
First, for each agent $i\in\mathcal{A}$ assigned to a goal $k\in\mathcal{F}$, Proposition \ref{prp:minE} implies that a finite arrival time, $t_{i,k}$ always exists.
Second, Propsition \ref{prp:solutionExistance} implies that a solution to the assignment problem (Problem \ref{prb:assignment}) always exists. 
This is sufficient to satisfy the premise of the Assignment Convergence Theorem presented in \cite{Beaver2020AnAgents}, which guarantees all agents arrive at a unique goal in finite time.
\end{proof}

\section{Optimal Path Planning} \label{sec:PathPlanning}
After being assigned to a goal with the optimal arrival time, each agent must find the energy-optimal trajectory to reach their assigned goal.
For trajectory generation, each agent plans over the horizon $[0, t_{i,k}]\subset\mathbb{R}_{\geq0}$, where $t=0$ is the current time and $t=t_{i,k}$ is the optimal arrival time.
The initial and final states of each agent $i\in\mathcal{A}$ is
\begin{align}
    \mathbf{p}_i(0) &= \mathbf{p}_i^0, \quad \quad &\mathbf{v}_i(0) &= \mathbf{v}_i^0, \label{eqn:init}\\
    \mathbf{p}_i(t_{i,k}) &= \mathbf{p}_i^a(t_{i,k}),  &\mathbf{v}_i(t_{i,k})&=\dot{\mathbf{p}}_i^a(t_{i,k}), \label{eqn:fin}
\end{align}
where $t_{i,k}$ is the argument that minimizes Problem \ref{prb:energy}.
To avoid collisions we impose a safety constraint to all agents with lower priority,
\begin{align}
    d_{ij}(t) \geq 2R,~~ &\forall j\in \{ \xi\in \mathcal{A} ~ | ~ \mathbb{I}_i (\xi) = 0\}, \label{eqn:safety} \\
    &\forall t \in [t_i^0, t_{i,k}]. \nonumber
\end{align}

Next, we formulate the decentralized optimal path planning problem.
\begin{problem}[Path Planning]  \label{prb:path}
    For each agent $i\in\mathcal{A}$ assigned to goal $k\in\mathcal{F}$, the optimal path can be found by solving the following optimal control problem,
    \begin{align}
         &\min_{\mathbf{u}_i(t)} \frac{1}{2} \int_{0}^{t_{i,k}} || \mathbf{u}_i(\tau)||^2 d\tau \\
         &\text{subject to: } \eqref{eqn:pDynamics}, \eqref{eqn:vDynamics}, \eqref{eqn:vBounds}, \eqref{eqn:uBounds}, \nonumber\\
         &\text{given: } \eqref{eqn:init}, \eqref{eqn:fin}. \nonumber
    \end{align}
\end{problem}

We derive the analytical solution to this problem by following the standard methodology used in optimal control problems with state and control constraints \cite{Malikopoulos2018a,Bryson1975AppliedControl,Malikopoulos2020,Ross2015}. First, we consider the unconstrained solution, given by \eqref{eqn:uUnc} - \eqref{eqn:pUnc}.
If the solution violates any of the constraints, then it is connected with the new arc corresponding to the violated constraint. This yields a set of the algebraic equation that are solved simultaneously using the boundary conditions of Problem \ref{prb:path} and interior conditions between the arcs. This process is repeated until no constraints are violated, which yields the feasible solution for Problem \ref{prb:path}.

The solution is a piecewise-continuous state trajectory composed of the following optimal motion primitives \cite{Beaver2020AnAgents}:
\begin{enumerate}
    \item no constraints are active,
    \item one safety constraint is active,
    \item multiple safety constraints are active,
    \item one state/control constraint is active, and
    \item multiple state/control constraint are active.
\end{enumerate}
For the full derivation of the solution for each case, see \cite{Beaver2020AnAgents}.

\section{Simulation Results} \label{sec:Simulation}

In this section, we present a series of simulation results to evaluate the effectiveness of the proposed method. All the simulations were conducted with $N=M=10$ agents and goals. The velocity of all the goals are given by the polynomials
\begin{equation}
    \mathbf{v}^*(t) = \left[\begin{array}{c}
         v_x^*(t)\\v_y^*(t)
    \end{array} \right] = \left[\begin{array}{cc}
         0.05t^3-0.3t^2+0.45t  \\
         0.02t+0.05
    \end{array}
    \right] .
\end{equation}
We randomly selected the initial positions of the agents in $\mathbb{R}^2$, which we then fixed for each simulation.

To demonstrate the effect of the energy-optimal arrival time (Problem \ref{prb:energy}), we compared the simulation results of the proposed method with that of the previous method \cite{Beaver2020AnAgents}, as shown in Fig.~\ref{fig:inf_proposed} and Fig.~\ref{fig:inf_previous}. 
We selected $T=5$ for the time parameter of the previous method. To remove the effect of decentralization on the performance, we set the sensing distance $h=\infty$ for both cases.

\begin{figure}[t!]
    \centering
    \includegraphics[width=0.8\linewidth]{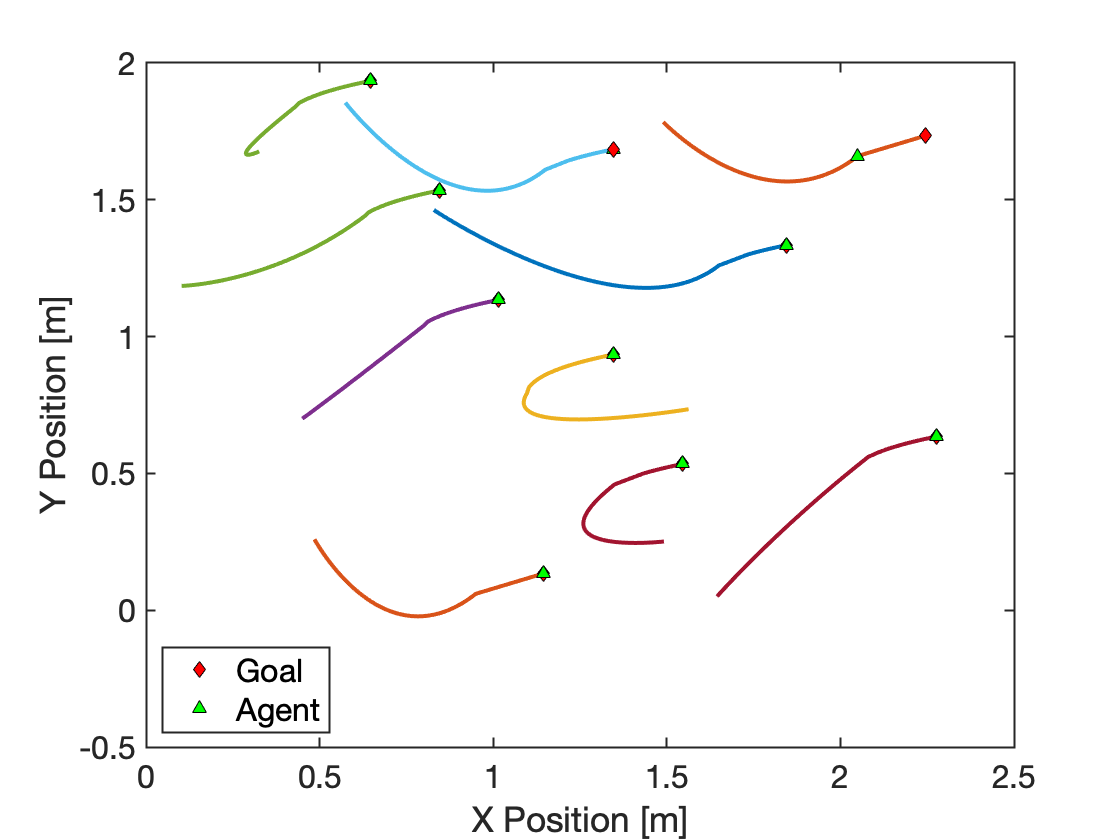}
    \caption{Simulation result for the proposed method with $h=\infty$}
    \label{fig:inf_proposed}
\end{figure}
\begin{figure}[t!]
    \centering
    \includegraphics[width=0.8\linewidth]{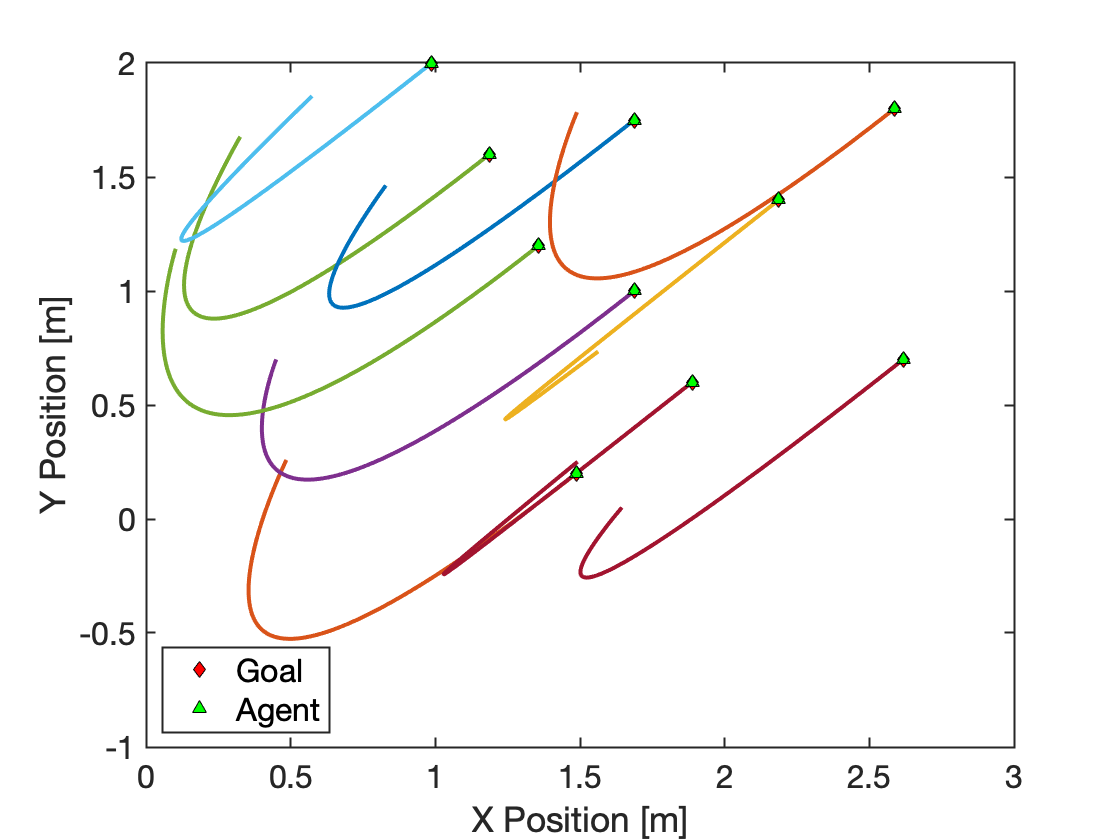}
    \caption{Simulation result for the previous method with $h=\infty$}
    \label{fig:inf_previous}
\end{figure}

\begin{figure}[tbh!]
    \centering
    \includegraphics[width=0.8\linewidth]{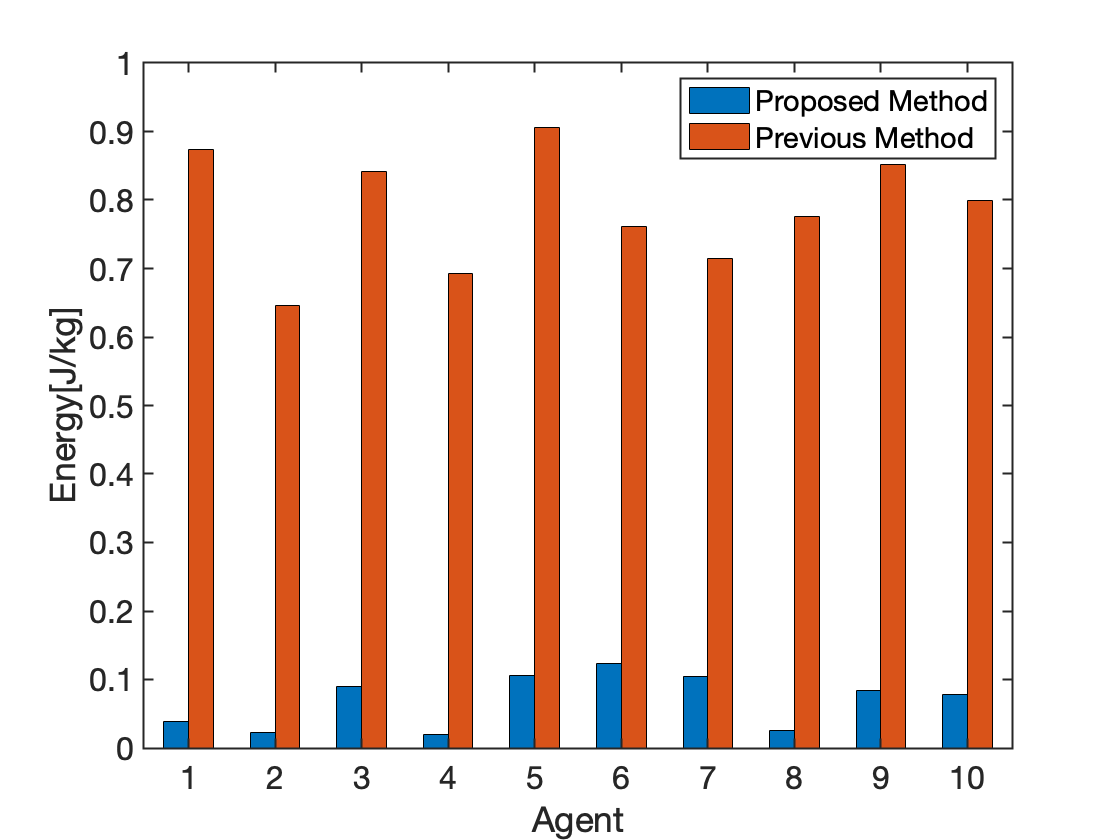}
    \caption{Total energy consumption of each agents}
    \label{fig:Etot}
\end{figure}

\begin{table}[tbh!]
    \centering
    \resizebox{0.95\textwidth}{!}{
    \begin{tabular}{c | c c}
         & Proposed method & Previous method  \\
         \hline
     Energy consumption &  0.69 kJ/kg & 7.86 kJ/kg\\
     Total arrival time & 4.57 s & 5 s
    \end{tabular}}
    \caption{Numerical result for comparison between the proposed method and the previous one.}
    \label{tab:comparison}
\end{table}

Numerical results are shown in Table \ref{tab:comparison}. The proposed method reduced the total energy consumption by 91.2\% compared to the previous method. 
This result shows that, in some cases, not only energy consumption but also the total time required to achieve the desired formation is improved.
We attribute this improvement to our algorithm selecting the optimal arrival time through Problem \ref{prb:energy}, rather than using a fixed arrival time.
The energy use of each agent for both cases are given in Fig. \ref{fig:Etot}, and all the agents consumed a minimum of 83.8\% to a maximum of 97.2\% less energy than the previous method.

\begin{table}[tb]
    \centering
    \begin{tabular}{c c c c c}
        \hline
        $h$ [m] & min. separation & $E$& $t_f$ & Total bans \\
        & [cm] & [kJ/kg] & [s] & \\
        \hline
        $\infty$ & 25.25 & 0.69 & 4.57  & 0\\
        1.25 & 16.84 & 454.3 & 4.57 & 6\\
        1.00 & 5.84 & 3.06 & 4.54 & 7\\
        0.75 & 10.62 & 15.48 & 4.54 & 5\\
        0.50 & 25.25 & 4.24 & 4.57 & 4\\
        \hline
    \end{tabular}
    \caption{Numerical results with different sensing distances.}
    \label{tab:data}
\end{table}

\begin{figure*}[th!]
    \centering
    \subfloat[$h=0.5$m]{\includegraphics[width=0.4\linewidth]{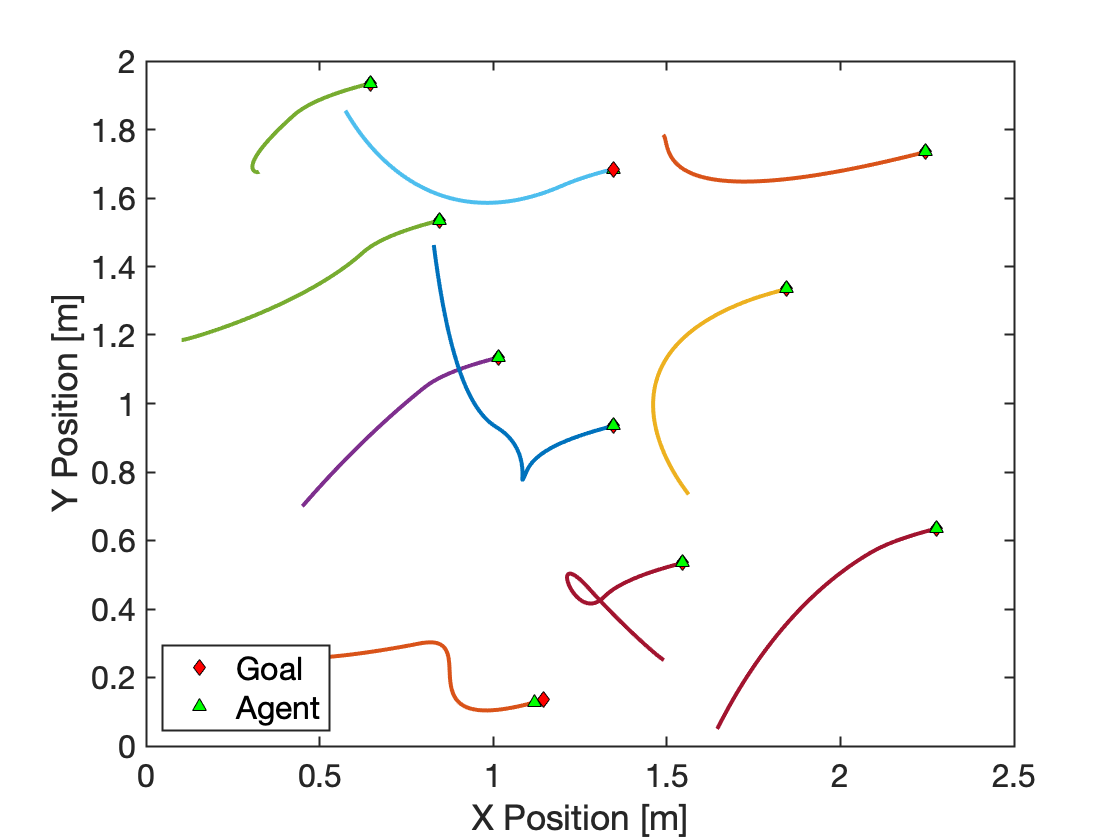}\label{a}}
    \subfloat[$h=0.75$m]{\includegraphics[width=0.4\linewidth]{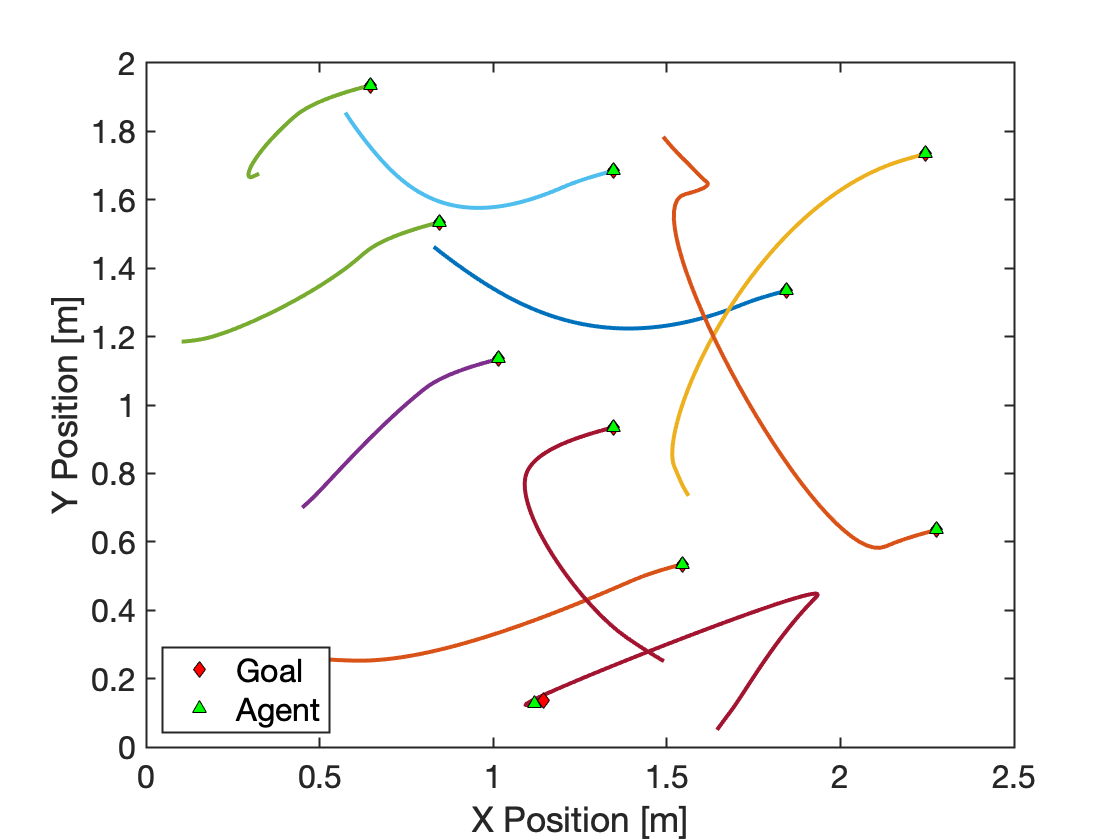}\label{b}}\\
    \subfloat[$h=1$m]{\includegraphics[width=0.4\linewidth]{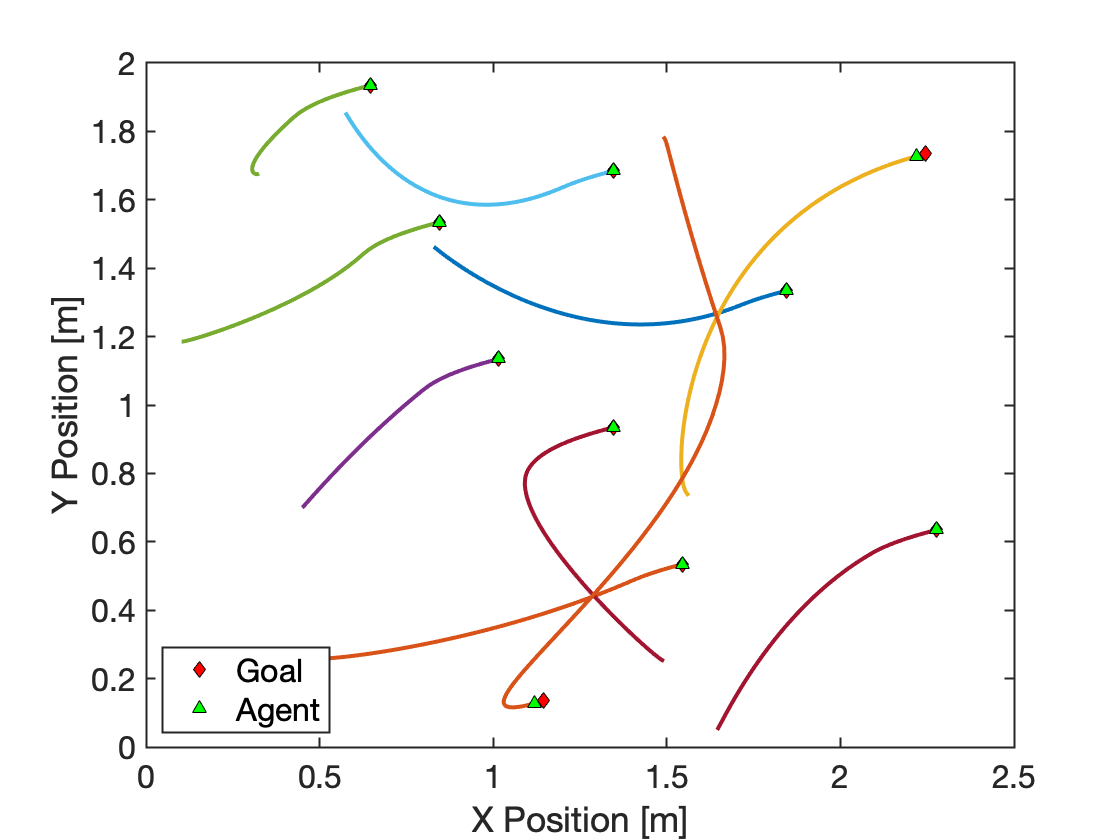}\label{c}}
    \subfloat[$h=1.25$m]{\includegraphics[width=0.4\linewidth]{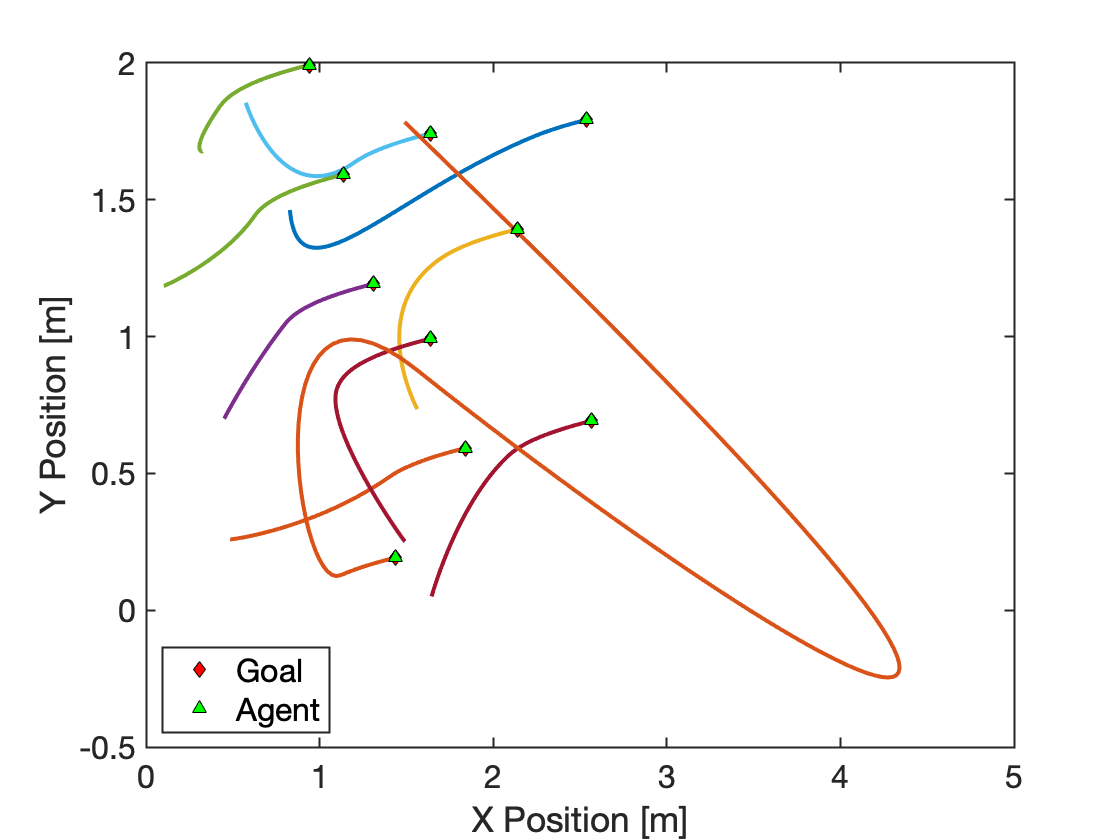}\label{d}}
    \caption{Trajectory of each agents with different sensing distances.}
    \label{fig:sensing}
\end{figure*}

Next, we simulated the agents with various sensing distances to understand its effect on performance.
We implemented a priority indicator function based on the neighborhood size, energy cost, and index of each agent as described in \cite{Beaver2020AnAgents}.
The results are shown in Table \ref{tab:data}, and Fig. \ref{fig:sensing} illustrates the trajectories generated by the agents 
with various values of $h$.
As with our previous work, \cite{Beaver2019AGeneration}, the results in Table \ref{tab:data} show no correlation between the sensing distance and energy consumption.
With respect to the agents' initial position and the desired formation, some information forces the agent to select the goal that is further than the one the agent would choose without that information, resulting in extra energy consumption.
This process is shown in Fig. \ref{fig:sensing}.
Compared to (a), the trajectory of one agent (shown with the orange line) gets longer and longer in (b), (c), and (d). 
The agent with a longer sensing distance may select a better goal at the beginning due to its extra information about other agents.
However, as shown in Table \ref{tab:data}, this may increase the number of banned goals, resulting in a higher number of assignments and reducing performance.

\section{Conclusion} \label{sec:Conclusion}

In this paper, we proposed an extension of our previous work on energy-optimal goal assignment and trajectory generation. The goal assignment task was separated into two sub-problems that include (1) finding energy-optimal arrival time and (2) assigning each agent to a unique goal. With the goal dynamics in the form of polynomials, we proved that our proposed approach guarantees that all agents arrive at a unique goal in finite time.
We validated the effectiveness of our approach through simulation.
Compared to previous work, we have shown a significant reduction in energy consumption.

Future work should consider how the initial position of the agents and desired formation affects energy consumption.
Quantifying the relationship between sensing distance and performance is another interesting area of research, as well as adapting agent memory and other information structures to the problem.
Finally, using recent results constraint-driven optimal control \cite{Beaver2020Energy-OptimalConstraints} to generate agent trajectories in real time is another compelling research direction.

\balance

\bibliographystyle{IEEEtran}
\bibliography{Bang, IDS, LoganMendeley}

\end{document}